\documentclass[a4paper,11pt]{article}
\RequirePackage{fix-cm}

\usepackage{ucs}
\usepackage[utf8x]{inputenc}
\usepackage[UKenglish]{babel} 
\usepackage{babelbib} 
\usepackage{amsmath} 
\usepackage{amsthm} 
\usepackage{amssymb}
\usepackage{graphicx}
\usepackage{pdfpages}
\usepackage{verbatim}
\usepackage{color}
\usepackage{paralist}
\usepackage{bbm}
\usepackage{url}
\usepackage{hyperref}
\usepackage{float}
\usepackage{mathtools}
\usepackage[normalem]{ulem}

\usepackage{url}

\allowdisplaybreaks
\newtheorem{theorem}{Theorem} 
\newtheorem{lemma}{Lemma}

\newtheorem{proposition}{Proposition}
\newtheorem{corollary}{Corollary}
\newtheorem{remark}{Remark}

\newcommand{\E}{\mathbb{E}}
\newcommand{\N}{\mathbb{N}}

\begin{document}





\title{On a dividend problem with random funding}

\author{Josef Anton Strini\thanks{Josef Anton Strini received support by the Austrian Science Fund (FWF) Single Project P26114.} \  \& Stefan Thonhauser
}
 \date{Graz University of Technology, Institute of Statistics\\ 
Kopernikusgasse 24/III, 8010 Graz, Austria\\
\vspace*{0.5cm}
18th {J}anuary 2019}


\maketitle
\begin{abstract}
We consider a modification of the dividend maximization problem from ruin theory. Based on a classical risk process we maximize the difference of expected cumulated discounted dividends and total expected discounted additional funding (subject to some proportional transaction costs). For modelling dividends we use the common approach whereas for the funding opportunity we use the jump times of another independent Poisson process at which we choose an appropriate funding height. In case of exponentially distributed claims we are able to determine an explicit solution to the problem and derive an optimal strategy whose nature heavily depends on the size of the transaction costs.
\end{abstract}

\section{Introduction and some first considerations}

\subsection{Overview}
In this article we deal with an extension of the classical dividend maximization problem for an underlying classical (compound Poisson) surplus process. Our proposed extension considers a random funding opportunity which is modelled by the following procedure. The insurer actively searches for investors who are willing to provide additional funding for the insurance portfolio under consideration. If the search is successful, the insurer can choose the height of the funding, increase the surplus and possibly pay out higher dividends in the future. We model the search procedure by means of an intensity $\beta\geq0$, such that the insurer finds funding opportunities at the jump-times of an additional and independent Poisson process. Naturally, new funding is costly or investors want to participate in future dividends respectively. That is why we weight this additional capital with a factor $\phi\geq 1$ which plays the role of a proportional transaction cost. The corresponding value function of our problem is the difference of expected cumulated discounted dividends and weighted expected cumulated discounted fundings, both up to the time of ruin.\\
For the case $\beta=0$ our approach just matches the classical dividend problem, i.e., no additional funding source can be found. Its treatment goes back to Gerber \cite{Gerb} and is analyzed in terms of optimal stochastic control by Azcue \& Muler \cite{AzMul2005DivRe,AzMulBook} and Schmidli \cite{schm}. The opposite extremal case $\beta \to \infty$ somehow resembles the situation of a possible capital injection at any point in time. This problem is by now well known under the keywords \emph{maximal dividends and capital injections} and was firstly formulated and solved by Kulenko \& Schmidli \cite{KuSchm} with the subtle difference that the controlled surplus process is not allowed to get ruined and thus resulting in a different value function.\\
Certainly, the approach of interventions at the jump times of another process is related to the formulation of ruin theoretic problems under random observations. Such a model comprising dividends is introduced by Albrecher et al. \cite{AlChTh} and gained some relevance in actuarial research over the last years. We need to emphasize that our present model is continuously monitored, i.e., dividend decisions can be made at any point in time and also the ruin event is immediately observed.\\
Another framework where dividend maximization problems for firm value determinations play a crucial role is finance. There the underlying process, typically given by a diffusion process, is interpreted as a cash reservoir of a company and the expected value of cumulated dividends reflects the value of this company. The present question is studied in a similar fashion in this \emph{financial diffusion} framework by Hugonnier et al. \cite{HuMM} in combination with an optimal stopping problem. As mentioned, the problem studied there is based on a continuous sample paths process and also the transaction costs parameter equals one, which results in a common single barrier type optimal strategy, both for dividends and fundings.\\
Interestingly, the recent  paper by Zhang et al. \cite{ZhChYa} study a compound Poisson risk model with a particular capital injection procedure, which is very similar to the optimal one derived in our contribution.
In contrast to our considerations, the focus is put on the determination of discounted penalty functions and dividend decisions are not part of the setup.\\
The paper is organized as follows. We start with the mathematical formulation of the model and associated stochastic optimization problem. In a next step we establish some basic properties of the value function and study parameter constellations which lead to degenerate optimal strategies. Having understood the crucial dependence on the magnitude of the transaction costs, we can subsequently determine the optimal strategy and corresponding value function. The key in this step is to prove the existence of a solution of a free-boundary value problem comprising two boundaries. We close the paper by some numerical illustrations which focus on the optimal strategy as a function of the transaction costs parameter $\phi$.

\subsection{Model setup}

In the subsequent lines we introduce the model of interest and the underlying stochastic protagonists. First of all we set up the stochastic basis of our considered model. We suppose a given  probability space $(\Omega,\,\mathbf{F},\,P)$ which carries the following underlying stochastic processes.\\
Let $N =(N_t)_{t\geq0}$ be a Poisson process with intensity $\lambda>0$ and let $\{Y_i\}_{i\in\N}$ be a sequence of independent and identically distributed random variables with distribution function denoted by $F_Y$ with $F_Y(0)=0$, we set $\E(Y_1)=\mu$ and assume $N$ to be independent of $\{Y_i\}_{i\in\N}$. Then we consider the following compound Poisson process $S=(S_t)_{t\geq0}$, 
\begin{equation*}
S_t = \sum_{i=1}^{N_t} Y_i,
\end{equation*}
which describes, as common in the classical risk model, the sum of all claims up to time $t$.
Next we consider a jump process $B = (B_t)_{t\geq0}$  with constant intensity $\beta$, i.e. a {P}oisson process,
with which we are able to describe the occurrence times of \emph{new} investors. In particular investors occur at the jump-times of $B$.
Again, independence between $N,\,B,\,\{Y_i\}_{i\in\N}$ is assumed.\\
Based on these ingredients we identify the filtration $\mathcal{F} = (\mathcal{F}_t)_{t\geq0}$ which models the available information at time $t$.
Consequently, we have to set
$$\mathcal{F}_t = \sigma \big \{ \mathcal{F}^{N}_t, \mathcal{F}^{B}_t, \{Y_1, Y_2, \dots, Y_{N_t}\} \big\} \cup \mathcal{N},$$
where $\mathcal{F}^{N}$ and $\mathcal{F}^{B}$ are the filtrations generated by the respective processes and $\mathcal{N}$ denotes the sets of measure zero.\\
Assuming that the insurance company has an initial surplus $x_0\geq0$ and receives premiums according to a rate $c$, we define the uncontrolled surplus or cash reserve process $X = (X_t)_{t\geq0}$, by
\begin{equation*}
X_t = x_0 + c t - S_t.
\end{equation*}
The control processes for the state process $X$ are on the one hand the dividend process $L = (L_t)_{t\geq0}$, an adapted and cáglád process, hence it is previsible, which is increasing and fulfills $L_0\equiv 0$. It represents the cumulated dividends up to time t.
On the other hand we consider the control process $f = (f_t)_{t\geq0}$, previsible as well, and $P$ - almost surely non-negative, i.e. $f_t \geq 0  \ P-a.s$. The control $f$ corresponds to the magnitude of the new funding at time $t$ in case $B$ jumps. According to that, the controlled cash reserve process reads as follows
\begin{equation*}
X^{L,f}_t = x_0 + ct - S_t - L_t + \int_0^t f_s dB_s.
\end{equation*}
In our setting it is not allowed that ruin is induced by dividend payments and therefore the relation
$$X^{L,f}_{t+}=X^{L,f}_t-\Delta L_t \geq 0,$$
has to hold $P-a.s.$
\begin{remark}
Due to the independence assumptions we have that $P-a.s.$ the two {P}oisson processes $N$ and $B$ do not jump at the same time. Since the paths of the dividend process $L$ are left-continuous one needs to read $\Delta L_t=L_{t+}-L_t$.
\end{remark}

\subsection{Optimization problem and value function}
The stated aim in our setting is to find the optimal combined dividend and funding strategy which maximizes the expected cumulated discounted future dividends deducting at least the received total additional funding.
The deduction depends on a proportional funding cost parameter denoted by $\phi\geq 1$. Hence the value function is defined by
\begin{equation*}
V(x) = \sup_{(L,f) \in \Theta} \E_x\left[ \int_0^{\tau^{L,f}} e^{-\delta t} dL_t - \phi \int_0^{\tau^{L,f}} e^{-\delta t} f_t dB_t \right ],
\end{equation*}
here $\tau^{L,f}$ denotes the first time when the controlled cash reserve process becomes negative, namely $\tau^{L,f}=\inf\{t\geq0|X^{L,f}_t < 0\}$ and
$\Theta$ is the set containing those admissible processes $(L,f)$ such that
$$\E_x\left[ \int_0^{\tau^{L,f}} e^{-\delta t} ( dL_t + \phi f_t dB_t) \right]<\infty.$$

If we assume that the controls are constant and the dividend control suffices $dL_t=ldt$ for some $l \in \mathbb{R}^+$,
we face a Markov process whose infinitesimal generator is
$$\mathcal{A}^{l,f}g (x) = cg'(x) - l g'(x)+\beta[g(x+f) - g(x)] + \lambda \int_0^\infty [g(x-y) - g(x)] dF_Y(y).$$

Naturally, the function $g$ above has to be in the domain of the generator $\mathcal{D}(\mathcal{A}^{l,f})$, which contains absolutely continuous functions $h$ satisfying an integrability condition $\E[\vert\,h(X^{l,f}_t)\,\vert]<\infty$, see Rolski et al. \cite[Th. 11.2.2]{rols}.
Using this expression we can state the Hamilton-Jacobi-Bellman equation of this problem
\begin{multline}\label{equ:HJB}
\max \bigg\{ c g'(x) - (\lambda  + \delta) g(x) + \lambda \int_0^x g(x-y) dF_Y (y)\\
+\beta\,\sup_{f\geq 0} \{ g(x+f) - g(x) - \phi f \} , 1- g'(x)\bigg\} =0.
\end{multline}

From the shape of the HJB-equation we can immediately derive some properties of its solutions.
\begin{lemma}\label{lem1}
Let $g$ be a continuously differentiable solution to the HJB-equation \eqref{equ:HJB}, then $g$ is strictly monotone increasing ($g'\geq 1$) and bounded from below by $\frac{c}{\lambda+\delta}>0$.
\end{lemma}
\begin{proof}
From the equation we directly obtain that $g'\geq 1$ and if we consider the limit
$x\searrow 0$ we get $g(0+) \geq \frac{c}{\delta+\lambda}>0$. Since $g$ is monotone increasing and continuous the assertion follows.
\end{proof}

Furthermore, we can bound the value function from below similarly as done by Azcue \& Muller \cite{AzMul2005DivRe} or by Schmidli \cite[p. 80 Lemma 2.37]{schm}.
\begin{lemma}
In the present model setup the value function fulfills 
$$V(x) \geq x + \frac{c}{\lambda+\delta}.$$
\end{lemma}

\begin{proof}
For the special choice $(L,f)\equiv (L,0)$ we face an admissible dividend strategy for the original dividend maximization problem.
The bound follows from the above cited (by now classical) results.
\end{proof}

\section{Solution of the optimization problem}

In the following we assume that the claim size distribution coincides with an exponential distribution with parameter $\alpha$. We try to identify an optimal strategy and determine an explicit solution to the problem. In case of an arbitrary claim size distribution one can expect a strategy of band type to be optimal. One needs to mention that the presence of the financing control complicates the situation in comparison to other modifications of the dividend problem with exponentially distributed claims in the literature.\\
We start with identifying parameter sets which lead to somehow degenerate optimal strategies.

\subsection{Optimality of keeping the reserve at zero}

For a special parameter configuration we obtain that the optimal strategy is to payout the initial reserve immediately
and keep on paying dividends such that the current reserve remains zero, which means that the dividend rate is $c$ and the first claim causes ruin.
Compare to classical results as presented in \cite[p. 93]{schm}.
\begin{lemma}
The optimal strategy is to payout immediately the initial reserve and then payout dividends at the premium rate $c$ 
if $(\delta+\lambda)^2\geq c\alpha\lambda$. Consequently, the value function has the following form 
\begin{equation}
V(x)= x+\frac{c}{\lambda+\delta}.
\end{equation}
\end{lemma}
\begin{proof}
Using $V(x)= x + \frac{c}{\lambda+\delta}$ the proof is analogous to the one given in \cite[p. 93]{schm}. Just note that in the present problem with capital supply the additional part of the
HJB-equation  corresponding to $f$ is zero,
$$\beta \sup_{f\geq 0} \{ V(x+f) - V(x) - \phi f \}= \beta \sup_{f\geq 0} \{(1 -\phi)f\}=0,$$
since $\phi \geq 1$.
\end{proof}
From the latter result we see that we need to focus on $(\delta+\lambda)^2 < c \alpha \lambda$,
which in turn implies that $c\alpha>\delta+\lambda$, since we assume that all parameters are positive.

\subsection{An embedded problem}
At the outset of tackling the problem we try as first conjectures some common types of controls such as barrier and \emph{simple} band strategies.
It turned out that they can not be optimal in general. Therefore, in order to get an idea of the shape of the optimal strategy we exploit a numerical approach.\\
At first fix $n\in\N$ and allow for at most $n$ capital injections (at the jump times $Z_1,\ldots,Z_n$ of $B$), the corresponding family of value functions is defined by 
\begin{align*}
V_n(x)=\sup_{(L,f)\in\Theta}\E\left[\int_0^{\tau^{L,f}}e^{-\delta s}dL_s-\phi\sum_{i=1}^{B_{\tau^{L,f}}\wedge n} e^{-\delta Z_i}f_{Z_i}\right].
\end{align*}
One may notice that $V_0$ is the value function of the classical dividend maximization problem and in the situation of exponentially distributed claims is explicitely known. However, the method below does not need to assume exponential claims but is focused on barrier type dividend strategies. Certainly, this can be generalized.\\
We have that $V_n$ after using one intervention \emph{restarts} with $V_{n-1}$, which can be used when maximizing with respect to $f$. The deduced numerical procedure is as follows:
\begin{itemize}
\item[1.] Compute $V_0$ without additional capital ($f=0$), by solving
\begin{align*}
0=c g'(x)-(\delta+\lambda)g(x)+\lambda\int_0^x g(x-y)dF_Y (y),
\end{align*}
for $0\leq x \leq b$ and $g'(x)=1$ for $x> b$. If the optimal $b$ is not known one can do this for different values of $b$.
Choosing the maximizing $b$, we obtain an approximation to $V_0$ with optimal barrier, say $b_0^*$.
Then we can compute the optimal state dependent $f_0(x)$ by setting
$$f_0=\underset{f\geq 0}{\mbox{argmax}}\{V_0(x+f)-V_0(x)-\phi f\}.$$
\item[2.] Compute $V_1$, where we allow for one financial injection, exactly $f_0$, and solve for different values of $b$
\begin{align*}
0=c g'(x)-(\delta+\lambda) g(x) + \lambda \int_0^x g(x-y) dF_Y (y)
+\beta[V_0(x+f_0)-V_0(x)-\phi f_0]
\end{align*}
for $0\leq x \leq b$ and $g'(x)=1$ for $x> b$. The usage of the maximizing $b$ results in an approximation of $V_1$.
As next step replace $V_0$ by $V_1$ in the first step and go on.
\end{itemize}
Of course we have to choose in every step the best value for the threshold $b$, which is illustrated in Figure \ref{fig:policyiteration}. There \emph{Vbstar}$(x)$ denotes the numerical solution of the usual dividend problem, further $V_i$ denotes the solution of the iteration, when $i \in \{1,2, \dots \}$ funding opportunities are allowed. The value inside the brackets correspond to the values of the barrier $b$, where $\Delta = 0.02$.
\begin{figure}[H]
\begin{center}
\includegraphics[scale=0.9]{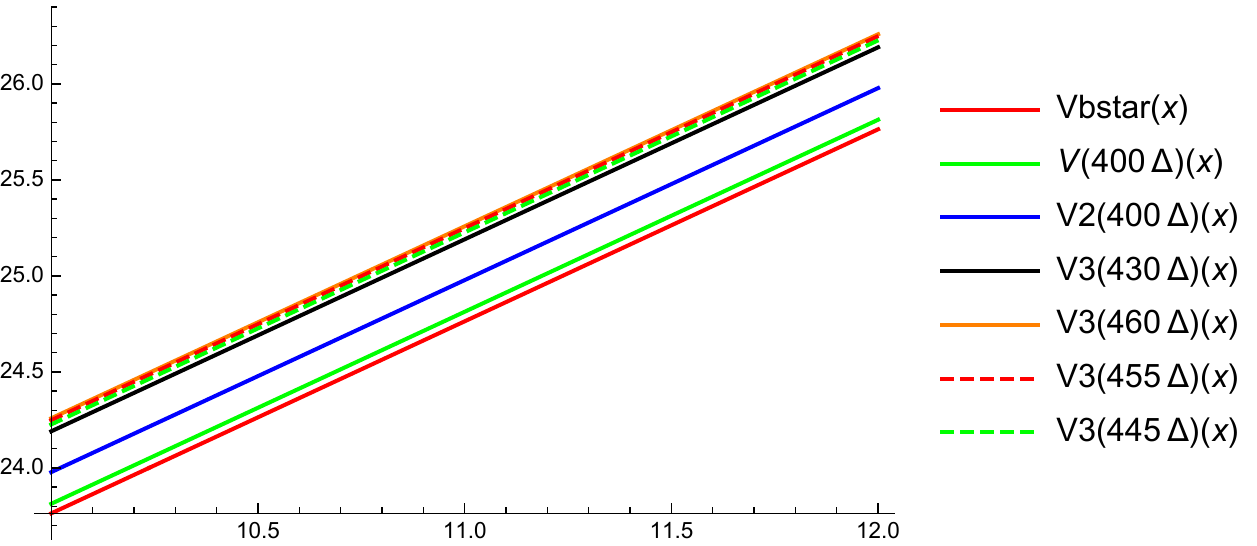}
\caption{Comparison of different iterative solutions.}\label{fig:policyiteration}
\end{center}
\end{figure}

This approach reveals a new type of possibly optimal strategy to us, which turns out to be the right conjecture.

\subsection{Resulting new strategy}
Using the results from the numerical approach, we are able to construct a new strategy for our problem.
The new strategy is of band type and specified by two parameters $0 \leq a \leq b<\infty$ such that
\begin{itemize}
\item the dividend strategy is of barrier type at level $b$,
\begin{align*}
\Delta L_{0+}=(x-b)I_{\{x>b\}}\\
dL_t=c\,dt,\quad t>0,
\end{align*}
\item the financing strategy only applies at reserve levels $x\in[0,a)$. It is given by $f(x)=(a-x)I_{\{0\leq x<a\}}$, with the feature that only below level $a$ we search for a funding source. If one appears, we choose the funding height to such an extent that the surplus jumps up to $a$ and in general not to the barrier level $b$.
\end{itemize}
For initial surplus $x\geq 0$ we denote the value, i.e., performance function, according to such a strategy by $V(x;a,b)$.
By construction it makes sense to write this function in the following form:
\begin{equation}\label{eq:funcstructure}
V(x;a,b) =
\begin{cases}
V_l(x;a,b), &\text{ if } 0 \leq x \leq a,\\
V_u(x;a,b), &\text{ if } a \leq x\leq b,\\
x-b + V(b;a,b)		, &\text{ if } x > b.
\end{cases}
\end{equation}
Whereby, using \emph{Dynkin-formula} type arguments or classical arguments based on \emph{conditioning on the first claim occurence}, the functions $V_l(x)=V_l(x;a,b)$ and $V_u(x)=V_u(x;a,b)$ have to fulfill the equations
\begin{multline}\label{equ:5}
cV_l'(x) - (\delta + \lambda)V_l(x) + \lambda
  \int_0^x V_l(x - y) \alpha e^{-\alpha y} dy \\
  + \beta (V_l(a) -V_l(x) - \phi (a-x))=0,
\end{multline}
\begin{multline}\label{equ:6}
cV_u'(x) - (\delta + \lambda)V_u(x) + \lambda
  \int_0^{x-a} V_u(x - y) \alpha e^{-\alpha y} dy \\
  + \lambda \int_{x-a}^x V_l(x - y) \alpha e^{-\alpha y} dy=0,
\end{multline}
\begin{equation}\label{equ:7}
V_l(a) = V_u(a) \text{ and } V_u'(b)=1.
\end{equation}
We get immediately, using the above equations, that $V(x;a,b)$ is continuously differentiable in $x$.
Furthermore, we obtain that continuity implies differentiability, i.e. $V'_l(a) = V'_u(a)$ if and only if $V_l(a) = V_u(a)$.
This means that the condition $V_l(a) = V_u(a)$ in \eqref{equ:7} is equivalent to the condition $V'_l(a) = V'_u(a)$.\\ 
We use the method of equating coefficients in order to solve the above equations \eqref{equ:5}, \eqref{equ:6} and \eqref{equ:7} explicitely.
This yields functions
\begin{align}
V_l(x;a,b) &\coloneqq A_1(a,b) e^{R_1 x} + A_2(a,b) e^{R_2 x}+A_3(a,b)x+A_4(a,b), \\
V_u(x;a,b) &\coloneqq B_1(a,b) e^{S_1 x} + B_2(a,b) e^{S_2 x},
\end{align}
where $S_1<0<S_2$ are solutions to $c S -(\delta+\lambda) +\frac{\alpha\lambda }{\alpha +S} =0$.
The exponents $R_1<0<R_2$ solve $c R-(\delta+\lambda+\beta) +\frac{\alpha  \lambda }{\alpha +R} =0.$
Note that under our assumptions we have that $S_1 + S_2 < 0$. The coefficients $A_1,\ldots,B_2$ are obtained by a system of five linear equations and do heavily depend on the parameters $a,\,b$.\\

We observe that numerical maximization of the function $V(x;a,b)$ in $a$ and $b$ for fixed $x$ results in levels $a^*$ and $b^*$ which are independent of $x$. These particular levels also coincide with the solutions of second order smooth fit conditions in $a$ and $b$. Additionally, we need to fulfill these second order smooth fit conditions in order to get a function which is twice continuously differentiable. This seems to be superfluous, since the domain of the generator only asks for absolut continuity, but in combination with concavity it serves as a basis for a direct proof of the associated verification theorem.  The conditions read as follows:
\begin{equation}\label{equ:8}
F(a,b)=
	\begin{pmatrix} 
	\frac{\partial^2}{\partial x^2}V_u(x;a,b)\big|_{x=b}\\
	\left[\frac{\partial^2}{\partial x^2}V_l(x;a,b)-\frac{\partial^2}{\partial x^2}V_u(x;a,b)\right]\big|_{x=a}
	\end{pmatrix}
	\stackrel{!}{=}
	\begin{pmatrix} 
	0\\
	0
	\end{pmatrix}.
	\end{equation}

In Figure \ref{fig:intersection} the green surface corresponds to $\frac{\partial^2}{\partial x^2}V_u(x;a,b)|_{x=b}$, the 
orange surface corresponds to $\frac{\partial^2}{\partial x^2}V_l(x;a,b)|_{x=a}-\frac{\partial^2}{\partial x^2}V_u(x;a,b)|_{x=a}$ and the blue plane highlights the zero level. The red curves mark the intersection of the two surfaces with the zero level.
\begin{figure}[htb]
\begin{center}
\includegraphics[scale=0.8]{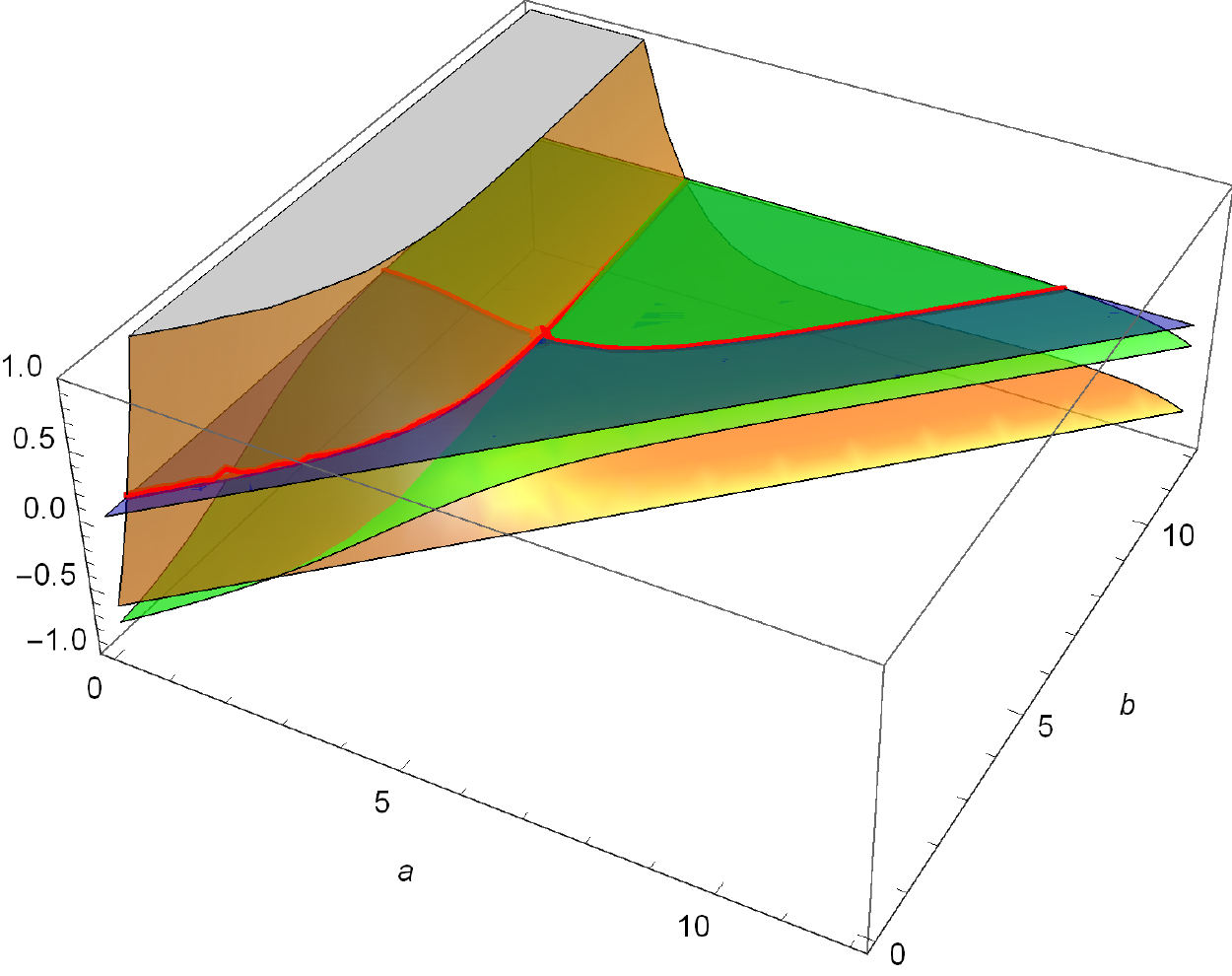}
\caption{Illustration of the problem of finding a solution $(a^*,b^*) $ to \eqref{equ:8}.}\label{fig:intersection}
\end{center}
\end{figure}

Finally, we have to prove the existence of thresholds $a^*$ and $b^*$ such that the smooth fit conditions are fulfilled. In our treatment we are able to derive an interesting condition, which turns out to be equivalent to one of the conditions above.\\
\\
If we differentiate the two integro-differential equations (for $x\in(0,a)$ and $x\in(a,b)$), which characterize $V_l$ and $V_u$, we obtain 
\begin{align*}
cV_l''(x) - (\delta + \lambda + \beta)V_l'(x) +  \lambda \alpha e^{-\alpha x} V_l(0) + \lambda  \int_0^{x} V'_l(x - y) \alpha e^{-\alpha y} dy + \beta \phi =0
\end{align*}
and
\begin{multline*}
cV_u''(x) - (\delta + \lambda)V_u'(x) + \underbrace{\lambda \alpha e^{-\alpha (x-a)} V_u(a) - \lambda \alpha e^{-\alpha (x-a)} V_l(a)}_{=0} \\
+ \lambda  \int_0^{x-a} V'_u(x - y) \alpha e^{-\alpha y} dy + \lambda \alpha e^{-\alpha x} V_l(0) + \lambda  \int_{x-a}^{x} V'_l(x - y) \alpha e^{-\alpha y} dy =0.
\end{multline*}
If we now set $x=a$ and calculate the difference of those two equations, we get using $V'_l(a) = V'_u(a)$ and $V_l(a) = V_u(a)$ that
\begin{equation*}
c(V_l''(a)-V_u''(a))= \beta( V_l'(a) - \phi).
\end{equation*}
Hence we have that $V_l''(a)=V_u''(a)$ if and only if $V_l'(a)= \phi$.\\
\begin{remark}
If we consider the following part of the HJB-equation: $$\sup_{f\geq 0} \{ V_l(x+f) - V_l(x) - \phi f \} ,$$
we obtain for the function $V_l(x)$, if it is concave, that the term inside the supremum is maximal if $V_l'(x+f)=\phi$. This means that using an $a$ such that $a=f+x=V_l'^{[-1]}(\phi)$, yields that the corresponding term ($V_l'(a) - \phi$) in the equation above is zero.
\end{remark}

\subsubsection{Extremal behaviour of optimal strategy}

First of all we want to know, whether the optimal strategy uses the additional funding or not, which means that we have to determine the reasons for the choice $a^*=0$. For that purpose we consider the solution of the usual dividend problem, which is well-known in the literature, see \cite{schm},

\begin{equation*}
\tilde{V}(x;b) =
		\begin{cases}
		\frac{h(x)}{h'(b)}, &\text{ if } 0 \leq x \leq b,\\
			x-b +\tilde{V}(b;b)		, &\text{ if } x > b.
		\end{cases}
\end{equation*}
Here
\begin{equation*}
h(x)= e^ {S_1 x} (S_1 + \alpha) - e^{S_2 x}(S_2 + \alpha),
\end{equation*}
where $S_1$ and $S_2$ are the same exponents as before and the optimal barrier ensuring twice continuous differentiability of the value function has the following form
\begin{equation*}
\tilde{b}= \ln\left(\frac{S_2^2 (S_2 + \alpha)}{S_1 ^2 (S_1 + \alpha)}\right)\frac{1}{S_1-S_2}.
\end{equation*}
In this case we have to assume that 
\begin{equation*}
(\delta + \lambda)^2 < c \alpha \lambda,
\end{equation*}
in order to make sure that $\tilde{b}>0$ holds. This yields that in the usual dividend problem $\tilde{V}(x;\tilde{b})$ is the value function.
We have that $\tilde{V}(x;b)= V(x;0,b)$.
\begin{remark}
At first one notices that
$$\tilde{V}'(0;\tilde{b}) = \frac{(S_1-S_2) (\alpha +S_1+S_2)}{S_1 (\alpha +S_1)
\left(\frac{S_2^2 (\alpha +S_2)}{S_1^2 (\alpha+S_1)}\right)^{\frac{S_1}{S_1-S_2}}-S_2 (\alpha +S_2)
\left(\frac{S_2^2 (\alpha +S_2)}{S_1^2 (\alpha+S_1)}\right)^{\frac{S_2}{S_1-S_2}}}\geq 1.$$
A second thought concerns the behaviour of an possibly optimal level $a$ with respect to the parameter $\phi$.
The following is more a heuristic intuition than a rigorous treatment. But if we consider $a=x+f^*=g'^{[-1]}(\phi)$, where $f^*$ is the maximizing argument of the supremum part of the HJB-equation
$$\sup_{f\geq 0} \{g(x+f) - g(x) - \phi f \} $$
as mentioned in the previous remark, we get (considering $a$ as a function of $\phi$, i.e.  $a(\phi)=g'^{[-1]}(\phi)$) that 
\begin{equation*}
a'(\phi)=(g'^{[-1]})'(\phi)=\frac{1}{g''(g'^{[-1]}(\phi))}<0,
\end{equation*}
if we assume that $g$ is concave, which will be true if $g$ is the value function. This means that the lower threshold $a$ is decreasing in $\phi$.
\end{remark}

\begin{proposition}
Let the optimal barrier in the classical dividend problem $\tilde{b}$ be positive.
For the optimal levels $0\leq a^* \leq b^*$ in the dividend problem with random funding we obtain that $a^* = 0$ and $b^*=\tilde{b}$ if and only if $\phi \geq \tilde{V}'(0;\tilde{b})$. 
\end{proposition}
This means that the simple barrier strategy is optimal and the solution of the classical dividend problem coincides with the solution of the extended problem.
\begin{proof}
We assume that $\phi \geq \tilde{V}'(0;\tilde{b})$ and have to show that $a^* = 0$ and $b^*=\tilde{b}$ are the optimal thresholds so that $ V(x;0,\tilde{b})=\tilde{V}(x;\tilde{b})$ solves the HJB-equation, is concave and $\mathcal{C}^2$ - the ingredients we later need in the verification theorem.\\
We know that this candidate function is concave, $\mathcal{C}^2$ and solves the HJB-equation of the classical dividend problem with $\beta=0$. Using this, we obtain for  $0\leq x < \tilde{b}$ that
$$1-\tilde{V}'(x;\tilde{b})<0.$$
For the first part of the HJB-equation it remains to show that $$\beta \sup_{f\geq 0} \{\tilde{V}(x+f;\tilde{b})-\tilde{V}(x;\tilde{b})-\phi f\} =0.$$
If we choose $f=0$ the term inside the supremum is zero. Otherwise, if $f>0$ we obtain that
$$\tilde{V}(x+f;\tilde{b})-\tilde{V}(x;\tilde{b})-\phi f<0.$$
This holds true because if we use $\tilde{V}''(x;\tilde{b})<0$ we get
$$\frac{\tilde{V}(x+f;\tilde{b})-\tilde{V}(x;\tilde{b})}{(x+f)-x} \leq \tilde{V}'(x;\tilde{b})< \tilde{V}'(0;\tilde{b})$$
and from $\phi \geq \tilde{V}'(0;\tilde{b})$, the above inequality follows.\\
For the values $x \geq\tilde{b}$ we know that $\tilde{V}'(x;\tilde{b})=1$, and that the first part of the HJB-equation (which coincides with the first part of the classical HJB-equation) is negative. So it remains to check whether
$$\beta \sup_{f\geq 0} \{\tilde{V}(x+f;\tilde{b})-\tilde{V}(x;\tilde{b})-\phi f\} \leq 0$$
holds. But this is true since if we plug in the linear function for $x > \tilde{b}$ we get that $\beta \sup_{f\geq 0} \{ (1-\phi)f\} \leq 0$, since $\phi \geq 1$. The special case $x=\tilde{b}$  is analogue to the situation $x< \tilde{b}$.\\
For the other direction, if $a^* = 0$ and $b^*=\tilde{b}$ are optimal, we have to show $\phi \geq \tilde{V}'(0;\tilde{b})$. For that purpose we assume the opposite, namely that $\phi < \tilde{V}'(0;\tilde{b})$.
But in this case we can exploit the fact that $\tilde{V}'(\tilde{b};\tilde{b})=1$ and $ \tilde{V}''(x;\tilde{b})<0 \ \forall x \in (0,\tilde{b})$, in addition to the above assumption $\tilde{V}'(0;\tilde{b}) > \phi \geq 1$. Putting this  together yields by the intermediate 
 value theorem that $\exists! \bar{x} \in (0, \tilde{b}): \ \tilde{V}'(\bar{x};\tilde{b})= \phi $. We want to show that $$\beta \sup_{f\geq 0} \{\tilde{V}(x+f;\tilde{b})-\tilde{V}(x;\tilde{b})-\phi f\} >0.$$ Differentiating the inner term and setting it equal to zero yields that $$\tilde{V}'(x+f;\tilde{b})\stackrel{!}{=}\phi,$$ so we obtain that $f^*= \bar{x}-x$ which is positive, if $ x  \in (0, \bar{x})$. Applying Taylor's formula gives for some $\theta \in (x,\bar{x})$ the following 
\begin{align*}
\tilde{V}(x;\tilde{b})=&\tilde{V}(\bar{x};\tilde{b})+ \tilde{V}'(\bar{x};\tilde{b})(x-\bar{x}) + \frac{1}{2}\tilde{V}''(\theta;\tilde{b})(x-\bar{x})^2\\
=&\tilde{V}(\bar{x};\tilde{b})- \phi (\bar{x} -x) + \frac{1}{2}\tilde{V}''(\theta;\tilde{b})(x-\bar{x})^2\\
<&\tilde{V}(\bar{x};\tilde{b})- \phi (\bar{x} -x).
\end{align*}
Finally we obtain that this function does not solve the HJB-equation and  $a^*=0$ being optimal cannot work out, which is a contradiction.
\end{proof}
The above proposition gives us the optimal strategy in the case $\phi \geq \tilde{V}'(0,\tilde{b})$. Furthermore, only in that case the usual dividend barrier strategy is optimal. In the next step we consider the lowest bound for the parameter $\phi$ where a non-trivial strategy appears namely $\phi = 1$.
\begin{lemma}
Let $\tilde{b}$ be positive. For the optimal levels $a^* \leq b^*$ for the dividend problem with random funding we obtain that $a^* = b^*$ if  $\phi =1$. 
\end{lemma}
In this case we are in the following situation, if an investor occurs we generate external funding to such an extent that we arrive with the surplus process at the dividend barrier, which triggers dividend payments. Hence, there is no gap between the dividend barrier and the funding level.
\begin{proof}
Solving the above equations for $a=b$, we get the function $V_l(x;a,a)$.
It remains to prove the existence of $a^*=b^*$, resulting in $V_l(x;a^*,b^*)$, such that the assumptions of the verification theorem are fulfilled.
At this point we know that $V'_l(a;a,a)=1$ and we have to find $a$ such that the smooth fit condition is fulfilled:
$M(a)\coloneqq V''_l(a;a,a)\stackrel{!}{=}0$.\\
Evaluating this function at $a=0$ yields
$M(0)=\frac{(\delta+\lambda)^2 - c\alpha \lambda}{c(\delta+\lambda)},$
which is negative according to our assumptions. Otherwise, we would have a value function of the form
$V(x)=x+\frac{c}{\delta+\lambda}$ as treated in the corresponding lemma above.
Furthermore, if $M(0)=0$ we obtain that $a^*=b^*=0$. This is also in line with the just mentioned case of a linear value function.\\
On the other hand we know that
$$M(a)=V''_l(a;a,a)= A_1(a,a) e^{R_1 a} + A_2(a,a) e^{R_2 a},$$
is continuous and $\lim\limits_{a \to \infty} M(a)=R_2 \frac{\delta}{\beta+\delta}$ is strictly positive.
This yields that there exists an $a^*$ such that $M(a^*)=0$. If there would be more than one point, such that the smooth fit conditions are fulfilled,
we decide to choose the smallest one.
At this $a^*>0$ we are able to exploit the equations $V'_l(a^*;a^*,a^*)=1$ and $V''_l(a^*;a^*,a^*)=0$ to get that $A_1(a^*,a^*)<0$ and $A_2(a^*,a^*)>0$. This yields that $V'''_l(x;a^*,a^*)> 0$ for all $ x \geq 0$, moreover together with $V''_l(a^*;a^*,a^*)=0$ we obtain that $V''_l(x;a^*,a^*)< 0$ for $0\leq x< a^*$. Which in turn implies that $V'_l(x;a^*,a^*)> 1$  for $0\leq x< a^*$. So the obtained function 
\begin{equation}
V(x)=V(x;a^*,a^*) =
\begin{cases}
V_l(x;a^*,a^*), &\text{ if } 0 \leq x \leq a^*,\\
x- a^* + V(a^*;a^*,a^*), &\text{ if } x > a^*.
\end{cases}
\end{equation}
is twice continuously differentiable and concave, in addition to that it fulfills the HJB-equation with $\phi=1$.
Altogether, this enables us to apply the verification theorem.
\end{proof}
%
%
\subsubsection{The case of moderate $\phi$}
Up to now we have investigated the following cases and obtained in each case the optimal combined strategy:
\begin{itemize}
\item  $\phi\geq \tilde{V}'(0;\tilde{b}) \Rightarrow a^*=0$ and $b^*=\tilde{b}$,
\item $\phi=1 \Rightarrow a^*=b^*$.
\end{itemize}
In this section we fill the missing gaps in order to have an optimal solution for every admissible value of $\phi\geq 1$.
\begin{theorem}
For $1<\phi<\tilde{V}'(0;\tilde{b})$ and $\tilde{b}>0$, we have that there exist $0<a^*<b^*$ such that the smooth fit conditions \eqref{eq:smoothfit1} and \eqref{eq:smoothfit2} are satisfied. The resulting function $V(x;a^*,b^*)$ from \eqref{eq:funcstructure} is a twice differentiable and concave solution to the HJB-equation \eqref{equ:HJB}.
\end{theorem}
\begin{proof}
Obviously, we have to solve the equations related to our band strategy for $1<\phi< \tilde{V}(0,\tilde{b})$.
This leads to a solution heavily depending on $a\leq b$ as in the previous case and it remains to choose the values $a$ and $b$ such that the equivalent smooth fit conditions are fulfilled, which are restated here
\begin{align}
V'_u(a)=B_1(a,b) S_1 e^{S_1 a} + B_2(a,b) S_2 e^{S_2 a}\stackrel{!}{=}\phi,\label{eq:smoothfit1}\\
V''_u(b)=B_1(a,b) S_1^2 e^{S_1 b} + B_2(a,b) S_2^2 e^{S_2 b}\stackrel{!}{=}0.\label{eq:smoothfit2}
\end{align}
Anyway, the coefficients $B_1(a,b)$ and $B_2(a,b)$ are fixed such that
\begin{equation}
V'_u(b)=B_1(a,b) S_1 e^{S_1 b} + B_2(a,b) S_2 e^{S_2 b}=1,
\end{equation}
holds true. Transforming this equation twice, leads to
\begin{align*}
B_1(a,b) S_1^2 e^{S_1 b}=S_1(1- B_2(a,b) S_2 e^{S_2 b}),\\
B_1(a,b) S_1 e^{S_1 a}=(1- B_2(a,b) S_2 e^{S_2 b})e^{S_1(a- b)}.
\end{align*}
Now we insert this expression into the equations \eqref{eq:smoothfit1} and \eqref{eq:smoothfit2} and obtain
\begin{align*}
B_2(a,b) S_2(e^{S_2 a} - e^{S_2 b +S_1 (a-b)})\stackrel{!}{=} \phi - e^{S_1 (a-b)},\\
B_2(a,b) S_2(e^{S_2 a} - e^{S_2 b +S_1 (a-b)})\stackrel{!}{=}\frac{(e^{S_2 a} - e^{S_2 b +S_1 (a-b)})S_1}{e^{S_2 b}(S_1-S_2)}.
\end{align*}
Note that according to our assumptions we have $S_1\neq S_2$. Combining those equations and rearranging terms results in
\begin{equation}\label{equ:10}
0\stackrel{!}{=} \phi (S_2-S_1)+S_1 e^{S_2 (a-b)}-S_2 e^{S_1 (a-b)}.
\end{equation}
For $h\geq0$ define $H(h):= \phi (S_2-S_1)+S_1 e^{-S_2 h}-S_2 e^{-S_1 h}$, then we have that $H(0)= ( \phi -1)(S_2-S_1)>0$, since $\phi>1$, $\lim\limits_{h\to \infty} H(h) = -\infty$ and $H'(h)<0$ for $h>0$, which means that there exists a unique $\bar{h}>0$ such that the equation is fulfilled.\\
Further it holds that if $1<\phi<\tilde{V}'(0,\tilde{b})$ then $0< \bar{h}< \tilde{b}$. Namely if there would exist an $h\geq \tilde{b}>0$ such that $H(h)=0$ then we would have 
$$\phi = \frac {S_2 e^{-S_1 h}-S_1 e^{-S_2 h}}{(S_2-S_1)}\geq  \frac {S_2 e^{-S_1 \tilde{b}}-S_1 e^{-S_2 \tilde{b}}}{(S_2-S_1)} = \tilde{V}'(0,\tilde{b}),$$ since the term on the left hand side of the inequality is strictly monotonically increasing in $h$ for all $h>0$. But this is a contradiction to the assumption for $\phi$. On top of this note that if $\phi=1$ or $\phi=\tilde{V}'(0,\tilde{b})$ then we obtain that $\bar{h}=0$ or $\bar{h}= \tilde{b}$ respectively, which is in line with the former investigations.\\
Finally, it remains to prove that for this given $\bar{h}$ there exists an $a$ such that $$V''_u(a+\bar{h};a,a+\bar{h})=B_1(a,a+\bar{h}) S_1^2 e^{S_1 (a+\bar{h})} + B_2(a,a+\bar{h}) S_2^2 e^{S_2 (a+\bar{h})}\stackrel{!}{=}0.$$
For this purpose we plug $\phi = \frac {S_2 e^{-S_1 \bar{h}}-S_1 e^{-S_2 \bar{h}}}{(S_2-S_1)}$ into this equation in order to work with the correct value for $h$. We obtain that 
$$V''_u(a+\bar{h};a, a+\bar{h})\Big|_{a=0}=\frac{S_1^2 e^{\bar{h} S_1} (\alpha +S_1)-S_2^2 e^{\bar{h} S_2} (\alpha +S_2)}{S_1 e^{\bar{h} S_1} (\alpha +S_1)-S_2 e^{\bar{h} S_2} (\alpha +S_2)}<0,$$
since $\bar{h}<\tilde{b}$. On the other hand, if we let $a$ tend to infinity we get that
$$\lim\limits_{a\to \infty}V''_u(a+\bar{h};a, a+\bar{h}):=C(\bar{h})>0.$$
This holds true since
$$C(\bar{h})=\frac{\alpha  P (\beta +\delta )+W S_1 S_2^2-e^{\bar{h} (S_1-S_2)} \left(\alpha  Q (\beta +\delta )+W S_1^2 S_2\right)}{\alpha  (\beta +\delta ) \left(\frac{P}{S_2}-e^{\bar{h} (S_1-S_2)}\frac{Q }{S_1}\right)},$$
where
\begin{align*}
P:=&S_2^2 (\alpha +S_2 ) (\alpha  \delta  (R_2-S_1) (\alpha +R_2+S_1)-\beta  S_1 (\alpha +R_2) (\alpha +S_1)),\\
Q:=&S_1^2 (\alpha +S_1) (\alpha  \delta  (R_2-S_2) (\alpha +R_2+S_2)-\beta  S_2 (\alpha +R_2) (\alpha +S_2)),\\
W:=&\beta (\alpha +S_1) (\alpha +S_2)
\left(\alpha ^2 (\beta +\delta )-\alpha  c R_2 (\alpha +R_2)+R_2^2 (\delta +\lambda )+\alpha R_2 (\beta +\delta +\lambda )\right).
\end{align*}
If we interpret $C$ as a function in $\bar{h}$ and evaluate it in zero, we see that $C(0)=\frac{\delta}{\beta+\delta}R_2>0$.
Using this and the continuity of $C(\bar{h})$ we obtain that $\exists\,\epsilon>0$ such that $C(\epsilon)>0$.
On top of this it even holds that $\frac{\partial}{\partial \bar{h}} C(\bar{h})=0$, which  implies that $C(\bar{h})>0$ for all $\bar{h}\geq 0$. Note that the denominator of $C(\bar{h})$ is strictly positive.\\
Finally, since the limit is positive, there exists an $a^*$ such that $V''_u(a^*+\bar{h})=0$, moreover note that $a^*>0$ since $\bar{h}<\tilde{b}$, if $\bar{h}=\tilde{b}$ then $a^*=0$. If there exists more than one $a^*$ such that this identity holds we decide to choose the smallest one, due to the shape of the function $V_l$. At this point, for $1<\phi< \tilde{V}'(0,\tilde{b})$, we know there exists $a^*>0$ and $b^*:=a^*+ \bar{h}>a^*$ such that the second order smooth fit conditions are fulfilled. Hence, the function
\begin{equation}\label{equ:11}
V(x;a^*,b^*) =
\begin{cases}
V_l(x;a^*,b^*), &\text{ if } 0 \leq x \leq a^*,\\
V_u(x;a^*,b^*), &\text{ if } a^* \leq x\leq b^*,\\
x-b^* + V(b^*;a^*,b^*)		, &\text{ if } x > b^*
\end{cases}
\end{equation}
is twice continuously differentiable.\\
As a next step we have to make sure that indeed our constructed function solves the HJB-equation and is concave.\\
First of all we obtain that for the coefficients of $V_u(x;a^*,b^*)$ it holds that $B_1(a^*,b^*)<0$ and $B_2(a^*,b^*)>0$. This is valid, since if we plug the equation
\begin{align*}
B_1(a^*,b^*) S_1^2 e^{S_1 b^*}=S_1(1- B_2(a^*,b^*) S_2 e^{S_2 b^*})
\end{align*}
into the equation $V''_u(b^*;a^*,b^*)=0$ and rearrange some terms, we obtain that $B_2(a^*,b^*)>0$, $B_1(a^*,b^*)<0$ follows analogously. This directly implies that $V'''_u(x;a^*,b^*)>0$ for all $x \geq 0$, together with $V''_u(b^*;a^*,b^*)=0$ we get that $V''_u(x;a^*,b^*)<0$ for all $x< b^*$, and this together with $V'_u(b^*;a^*,b^*)=1$ yields that $V'_u(x;a^*,b^*)>1$ for all $a^*\leq x < b^*$.\\
Furthermore, the first coefficient of $V_l(x;a^*,b^*)$ satisfies that $A_1(a^*,b^*)<0$, since we can use the identity
\begin{align*}
A_2(a^*,b^*) R_2^2 e^{R_2 a^*}=R_2(\phi - A_1(a^*,b^*) R_1 e^{R_1 a^*}-A_3(a^*,b^*)),
\end{align*}
in order to derive from the inequality $V''_l(a^*;a^*,b^*)=V''_u(a^*;a^*,b^*)<0$ that $A_1(a^*,b^*)<0$. Knowing that $A_3(a^*,b^*)>0$,
we distinguish between the following cases, namely if $A_2(a^*,b^*)<0$ then $V''_l(x;a^*,b^*)<0$ for all $0\leq x \leq a^*$ and this property together with $V'_l(a^*;a^*,b^*)=\phi>1$ yields that $V'_l(x;a^*,b^*)\geq \phi >1$ for all $0\leq x \leq a^*$.\\
If $A_2(a^*,b^*)>0$, then $V'''_l(x;a^*,b^*)>0$ for all $0\leq x \leq a^*$, which implies that $V''_l(x;a^*,b^*)<0$ for all $0\leq x \leq a^*$, since $V''_l(a^*;a^*,b^*)=V''_u(a^*;a^*,b^*)<0$.
Now the concavity together with $V'_l(a^*;a^*,b^*)=\phi>1$ yields that $V'_l(x;a^*,b^*)\geq \phi >1$ for all $0\leq x \leq a^*$.
Additionally, we can deduce that $V(x;a^*,b^*)>\frac{\phi c}{\delta+\lambda}>0$.
This can be shown as already done in Lemma \ref{lem1}, just by using the equation for $V_l(x;a^*,b^*)$ in $x=0$ and exploiting that
$\beta(V_l(a^*;a^*,b^*)-V_l(0;a^*,b^*)-\phi a^*)>0$, due to concavity and $V'_l(a^*;a^*,b^*)=\phi$.\\
In the end, if we insert the function $V(x;a^*,b^*)$ into the HJB-equation we obtain that for $x \in [0,a^*]$ the first part of the HJB-equation is zero and the second part is less than zero. For $x \in (a^*,b^*]$ the same holds true, since the supremum term in the first part is zero, because
\begin{align*}
V(x+f;a^*,b^*) =& V(x;a^*,b^*) + V'(x;a^*,b^*)f + V''(\theta;a^*,b^*)\frac{1}{2} f^2
<&V(x;a^*,b^*)+ \phi f
\end{align*}
holds true, for a $\theta \in (x,x+f)$, provided that $f>0$, otherwise if $f=0$ the supremum part is also zero.\\
For $x>b^*$ we have to show that the second part of the HJB-equation is zero and the first part is less than zero.
For that reason we consider the function 
\begin{multline*}
q(x):= c V'(x;a^*,b^*)- (\lambda +\delta) + \lambda \int\limits_0^x V(x-y;a^*,b^*)\alpha e^{- \alpha y} dy\\ + \beta \sup_{f\geq 0} \{ V(x+f;a^*,b^*) - V(x;a^*,b^*) - \phi f\},\quad x>b^*.
\end{multline*}
We have to show that $q(x)<0$ for all  $x > b^*$. We already know that $q(b^*)=0$ and that the supremum part in $q(x)$ is zero,
since $V(x;a^*,b^*)$ is linear for $x>b^*$ and $\phi>1$.
Furthermore, we use the properties of the coefficients of $V_l(x;a^*,b^*)$ and $V_u(x;a^*,b^*)$. Together with the smooth fit conditions \eqref{eq:smoothfit1} and \eqref{eq:smoothfit2} we get, surprisingly nice,
$$B_1(a^*,b^*)=\frac{S_2\,\phi\,e^{b^* S_2}}{S_1 S_2 e^{a^* S_1 +b^* S_2}-S^2_1  e^{a^* S_2+b^* S_1 }},$$
$$B_2(a^*,b^*)=\frac{S_1\,\phi\,e^{b^* S_1 }}{S_1  S_2  e^{a^* S_2+b^* S_1}-S_2^2 e^{a^* S_1+b^* S_2}}.$$
In addition to that, we use the identity for $\phi$ given in \eqref{equ:10} to obtain 
\begin{equation}
q'(x)= (e^{(b^*-x)\alpha}-1)\delta<0,\;\text{for all}\;x>b^*.
\end{equation}
Finally, this yields that $q(x)<0$ for $x>b^*$, which verifies that $V(x;a^*,b^*)$ satisfies the first part of the HJB-equation for $x>b^*$. Obviously, the function $V(x;a^*,b^*)$ satisfies $1-V'(x;a^*,b^*)=0$ for $x>b^*$ and this shows that $V(x;a^*,b^*)$ solves the second part of the HJB-equation. Overall, this means that the function specified in \eqref{equ:11} solves the HJB-equation \eqref{equ:HJB}.
\end{proof}

\section{Verification Theorem}
Here we state a \emph{verification theorem} which fits to our constructed function $V(x;a^*,b^*)$ in \eqref{equ:11}.
\begin{theorem}
Let $g \in \mathcal{C}^2(0,\infty)$ be a positive solution to the HJB-equation
\begin{multline*}
\max \bigg\{ c(x) g'(x) - (\lambda + \delta) g(x) + \lambda \int_0^x g(x-y) dF_Y (y)\\ +  \beta \sup_{f\geq 0} \{ g(x+f) - g(x)-  \phi f \} , 1- g'(x)\bigg\} =0.
\end{multline*} 
We set $g(x)=0$, if $x<0$. Further let $g$ be concave, then
$$g(x) \geq V(x),$$
where $$V(x) = \sup_{(L,f)\in\Theta} \E_x\left[ \int_0^{\tau^{L,f}} e^{-\delta t} dL_t - \phi \int_0^{\tau^{L,f}} e^{-\delta t} f_t dB_t \right ]$$
and $\tau^{L,f}=\inf\{t\geq 0 | X^{L,f}_t <0\}$.
\end{theorem}

\begin{proof}
Let $g \in C^2(0,\infty)$ and $(L,f)$ be an admissible control strategy. In the following we will denote the state process $X^{L,f}_t$ depending on $(L,f)$ with $X_t$ and $\tau^{L,f}$ with $\tau$ for the sake of clarity.
Because we want to make use of important theorems from stochastic calculus we have to switch to the right-continuous process, see also Shreve et al. \cite[p. 60-62]{Shre}. We consider the process
\begin{equation}\label{eq:Y}
Y_t = e^{-\delta (t \wedge \tau)} g(\bar{X}_{t \wedge \tau}) + \int_{0}^{t \wedge \tau} e^{-\delta s} dL_{s+} -  \phi \int_{0}^{t \wedge \tau} e^{-\delta s} f_s dB_s,
\end{equation}
where $\bar{X}_t\coloneqq X_{t+}$.
First of all we apply the integration by parts formula to the first part of $Y$ and It$\hat{\text{o}}$'s formula for $\bar{X_s}$. We get
\begin{multline*}
e^{-\delta ( t \wedge \tau)} g(\bar{X}^\tau_{t})=  g(X_{0+}) + \int_{0+}^{t\wedge \tau} e^{-\delta s} [ -\delta g(\bar{X}_{s-})+ c g'(\bar{X}_{s-})]ds \\
- \int_{0+}^{t\wedge \tau} e^{-\delta s} g'(\bar{X}_{s-})dL^c_s +\sum_{0<s\leq t \wedge \tau} e^{-\delta s} \Delta g(\bar{X}_{s-}).
\end{multline*}
Moreover, we can split up the above sum of the discontinuous parts such that we obtain
\begin{align*}
\sum_{0<s\leq t \wedge \tau} e^{-\delta s} \Delta g(\bar{X}_{s-})=&\sum_{0<s \leq t\wedge \tau , \ L_{s+}\neq L_s} e^{-\delta s} \left[ g(\bar{X}_{s-} - \Delta L_{s+}) - g(\bar{X}_{s-})\right] \\
&+\sum_{0<s \leq t\wedge \tau , \ S_{s}\neq S_{s-}} e^{-\delta s} \left[ g(\bar{X}_{s-} - Y_{N_s}) - g(\bar{X}_{s-}) \right] \\
&+\sum_{0<s \leq t\wedge \tau , \ B_{s}\neq B_{s-}} e^{-\delta s} \left[ g(\bar{X}_{s-} + f_s) - g(\bar{X}_{s-})\right].
\end{align*}
As in \cite[p. 19 - 20]{AzMu} with $g' \geq 1$, we obtain for the sum belonging to the jumps of the dividend process the estimate 
\begin{multline*}
\sum_{0<s \leq t\wedge \tau , \ L_{s+}\neq L_s} e^{-\delta s} \left[ g(\bar{X}_{s-} - \Delta L_{s+}) - g(\bar{X}_{s-})\right] = \\
-\sum_{0<s \leq t\wedge \tau , \ L_{s+}\neq L_s} e^{-\delta s} \int_{0+}^{\Delta L_{s+}} g'(\bar{X}_{s-}-u)du\ \leq - \sum_{0<s \leq t\wedge \tau , \ L_{s+}\neq L_s} e^{-\delta s} \Delta L_{s+}.
\end{multline*}
For the sums with comprising the other jumps we take expectations and use the compensation formula to get
\begin{multline*}
\E_x\left[\sum_{0<s \leq t\wedge \tau , \ S_{s}\neq S_{s-}} e^{-\delta s} \left[ g(\bar{X}_{s-} - Y_{N_s}) - g(\bar{X}_{s-}) \right] \right]= \\
\E_x\left[ \int_{0+}^{t \wedge \tau} e^{-\delta s}  \left(\lambda \int_0^{\bar{X}_{s-}} g(\bar{X}_{s-} -y)dF_Y(y)-\lambda g(\bar{X}_{s-})   \right)ds \right]
\end{multline*}
and 
\begin{multline*}
\E_x\left[\sum_{0<s \leq t\wedge \tau , \ B_{s}\neq B_{s-}} e^{-\delta s} \left[ g(\bar{X}_{s-} + f_s) - g(\bar{X}_{s-})\right] \right]= \\ 
\E_x\left[  \int_{0+}^{t \wedge \tau} e^{-\delta s} \beta \left[g(\bar{X}_{s-} + f_s)-g(\bar{X}_{s-}) \right] ds\right].
\end{multline*}
Now, we exploit the results from above to obtain
\begin{align*}
\E_x\left[  e^{-\delta ( t \wedge \tau)} g(\bar{X}^\tau_{t})\right] \leq&
\E_x\Bigg[ g(X_{0+}) + \int_{0+}^{t\wedge \tau} e^{-\delta s} [ -\delta g(\bar{X}_{s-})+ c g'(\bar{X}_{s-})]ds\\
 &+ \int_{0+}^{t\wedge \tau} e^{-\delta s} \left( \lambda \int_{0}^{\bar{X}_{s-}} g(\bar{X}_{s-} -y)dF_Y(y)-\lambda g(\bar{X}_{s-})  \right) ds  \\
 &+\int_{0+}^{t\wedge \tau} e^{-\delta s}\beta \left[g(\bar{X}_{s-} + f_s)-g(\bar{X}_{s-})  \right]  ds - \int_{0+}^{t\wedge \tau} e^{-\delta s} dL_{s+} \Bigg].
\end{align*}
Next we use that $g$ solves the HJB-equation
\begin{align*}
\E_x\left[  e^{-\delta ( t \wedge \tau)} g(\bar{X}^\tau_{t})\right] \leq \E_x\Bigg[ g(X_{0+})- \int_{0+}^{t\wedge \tau} e^{-\delta s} dL_s + \int_{0+}^{t\wedge \tau} e^{-\delta s}\beta \phi f_s ds  \Bigg].
\end{align*}
Adding on both sides $\E_x[ \int_{0}^{t\wedge \tau} e^{-\delta s} dL_{s+} - \int_{0}^{t\wedge \tau} e^{-\delta s}\beta \phi f_s ds]$ and using the concavity of $g$ yields that
\begin{align*}
\E_x\left[ Y_t \right] \leq \E_x\Bigg[ g(X_{0+})+\Delta L_{0+}   \Bigg] \leq g(x).
\end{align*}
The last inequality yields that the process $Y$ is a supermartingale. Now we use this property to obtain
\begin{align*}
 g(x)= Y_0 \geq \E_x(Y_t) \geq \E_x\left[ \int_{0}^{t \wedge \tau } e^{-\delta s}( dL_s - \phi f_s dB_s)  \right] \geq \\
 \E_x\left[ \int_{0}^{\lfloor t \rfloor \wedge \tau } e^{-\delta s} dL_s \right] -  \E_x\left[ \int_{0}^{(\lfloor t \rfloor +1) \wedge \tau } e^{-\delta s}  \phi f_s dB_s  \right],
\end{align*}
where we exploited that $g\geq 0$. Considering the limit $t \to \infty$ and using monotone convergence yields that
\begin{align*}
 g(x)\geq \E_x\left[ \int_{0}^{\tau } e^{-\delta s}( dL_s - \phi f_s dB_s)  \right],
\end{align*}
taking the supremum over all admissible strategies gives the desired result:
\begin{align*}
 g(x)\geq V(x).
\end{align*}
\end{proof}
Since for all parameter constellations our constructed functions are linked to an admissible strategy, are twice differentiable and concave, we have that they dominate the value function. Furthermore, using the band type strategy specified by $(a^*,b^*)$ we have that the corresponding $Y$ from \eqref{eq:Y} is a martingale. Instead of using dominated convergence in the limitation procedure, one can even use bounded convergence, since $f^*_s\leq a^*$ and observe that
$\lim_{t\to\infty}\E_x\left[e^{-\delta (t\wedge\tau)}V(\bar{X}_t^\tau;a^*,b^*)\right]=0$.
\begin{corollary}
If $\tilde{b}>0$ and $\phi\geq 1$, the function $V(x;a^*,b^*)$ is the value function and the corresponding band type strategy is optimal.
\end{corollary}

\section{Numerical illustration}

In this concluding section we present a numerical example which nicely illustrates the dependence of the optimal strategy on the parameter $\phi\geq 1$.
For this purpose we have chosen the parameters as follows. Concerning the reserve process we take $c = 1.5$ for the premium rate, $\lambda = 1$ for the intensity of the Poisson process corresponding to the claims, $\alpha = 1.5$ for the parameter of the exponential distribution of the claim size.
Furthermore, for the jump process $B$ we take $\beta = 2$, which corresponds to the expected arrivals of investors per time unit.
In terms of the interest rate we choose $\delta = 0.02$ and in order to illustrate the value function and the smooth fit conditions we fix $\phi = 1.5$ temporarily.\\
In Figure \ref{fig:ValFunc} we depict the difference between the value function of the usual dividend problem and the value function of the model with random capital supply with $a^* = 3.1746,\,b^*= 6.8526$. Figure \ref{fig:strats} illustrates how the transaction cost parameter $\phi$ affects the nature of the optimal strategy in terms of $(a^*,b^*)$. As proved above, we observe that for the case $\phi=1$ the two thresholds $a^*$ and $b^*$ coincide. Further, if $\phi$ increases, the area where we search for additional funding shrinks exactly up to the certain point where it disappears.
This exactly happens at $\phi =\tilde{V}'(0;\tilde{b})$. Simultaneously, the dividend threshold $b^*$ is increasing in $\phi$ and reaches its maximum level at the point where $a^*$ becomes zero, namely, again if $\phi =\tilde{V}'(0;\tilde{b})$. We observe that the maximum level for $b^*$ is the dividend barrier level $\tilde{b}$ of the usual dividend problem.
On top of this we even see (and indeed proved) that for values of $\phi$ larger than $\tilde{V}'(0;\tilde{b})$ the optimal strategy does not change anymore.

\begin{figure}[htb]\centering
\begin{minipage}{0.9\linewidth}
\includegraphics[width=\linewidth]{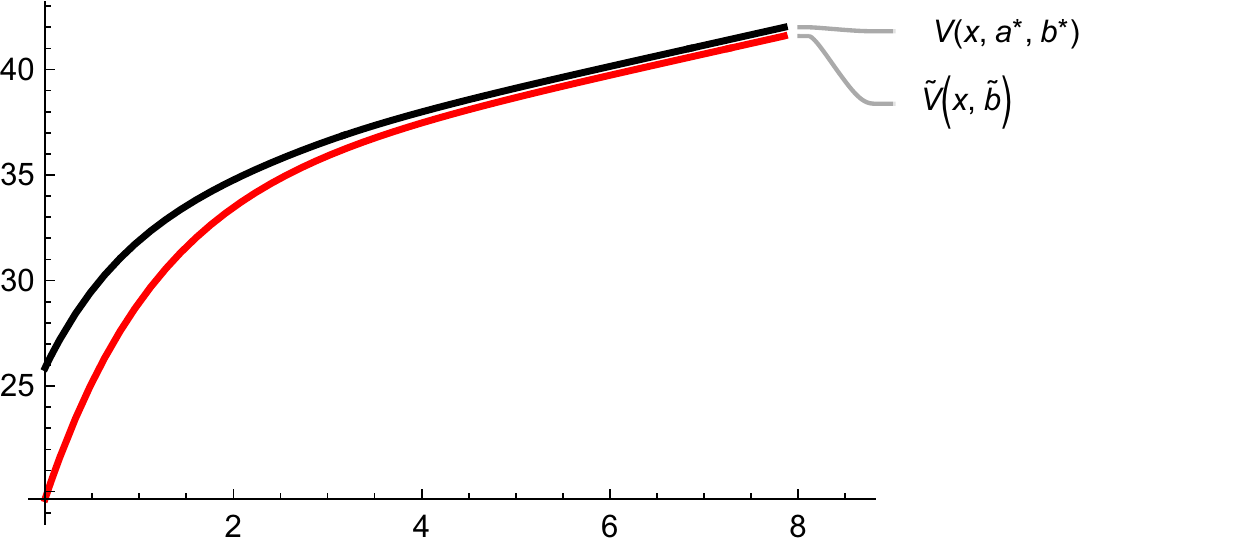}
\caption{Value functions}\label{fig:ValFunc}
\vspace*{0.5cm}
\end{minipage}
\begin{minipage}{0.9\linewidth}
\includegraphics[width=\linewidth]{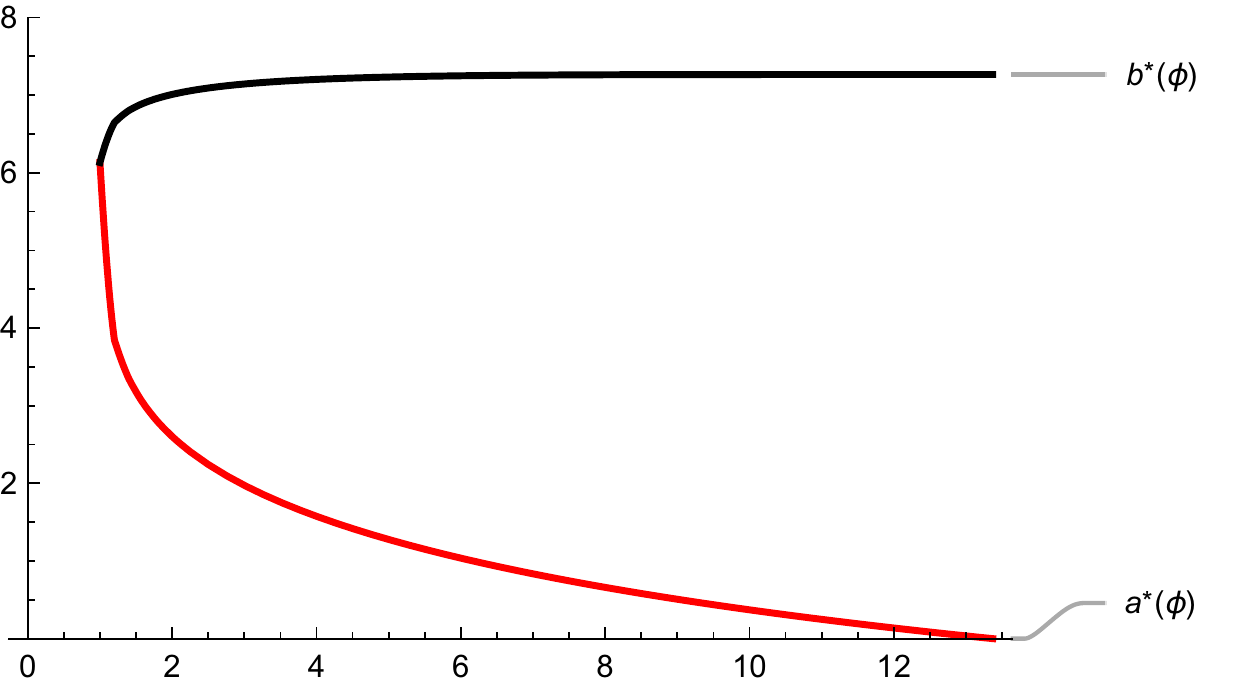}
\caption{Strategies as functions of $\phi$}\label{fig:strats}
\end{minipage}
\end{figure}

In the Figures \ref{fig:1stsmooth} and \ref{fig:2ndsmooth} we illustrate the first and second order smooth fit property. In Figure \ref{fig:1stsmooth} we plotted the first derivative of the value function to point out that at the lower optimal threshold $a^*$ we have  $V'(a^*;a^*,b^*)= \phi$, which is, according to our theoretical treatment, equivalent to the second order smooth fit condition. Further, at the upper optimal threshold $b^*$ we have that $V'(b^*;a^*,b^*)= 1$. Finally, Figure \ref{fig:2ndsmooth} shows the second derivative of the functions $V_l(x;a^*,b^*)$ and $V_u(x;a^*,b^*)$ and illustrates their behaviour in the respective domain of interest.

\begin{figure}[htb]\centering
\begin{minipage}{0.9\linewidth}
\includegraphics[width=\linewidth]{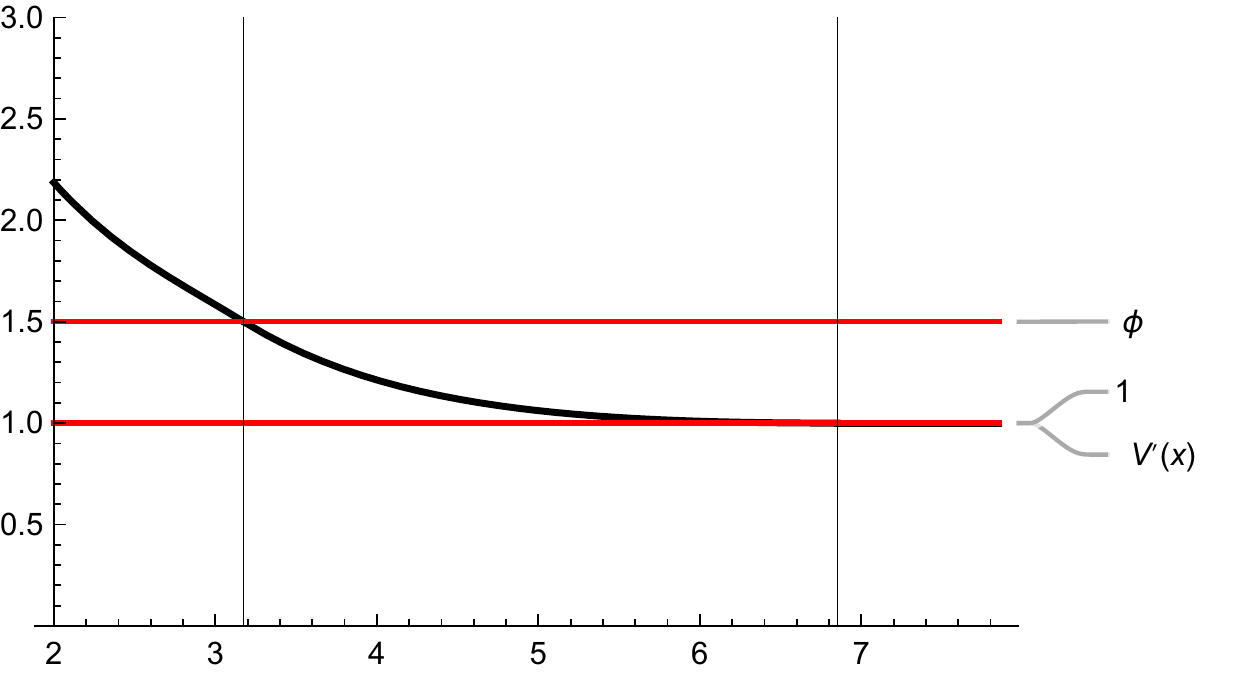}
\caption{1st order smooth fit}\label{fig:1stsmooth}
\vspace*{0.5cm}
\end{minipage}

\begin{minipage}{0.9\linewidth}
\includegraphics[width=\linewidth]{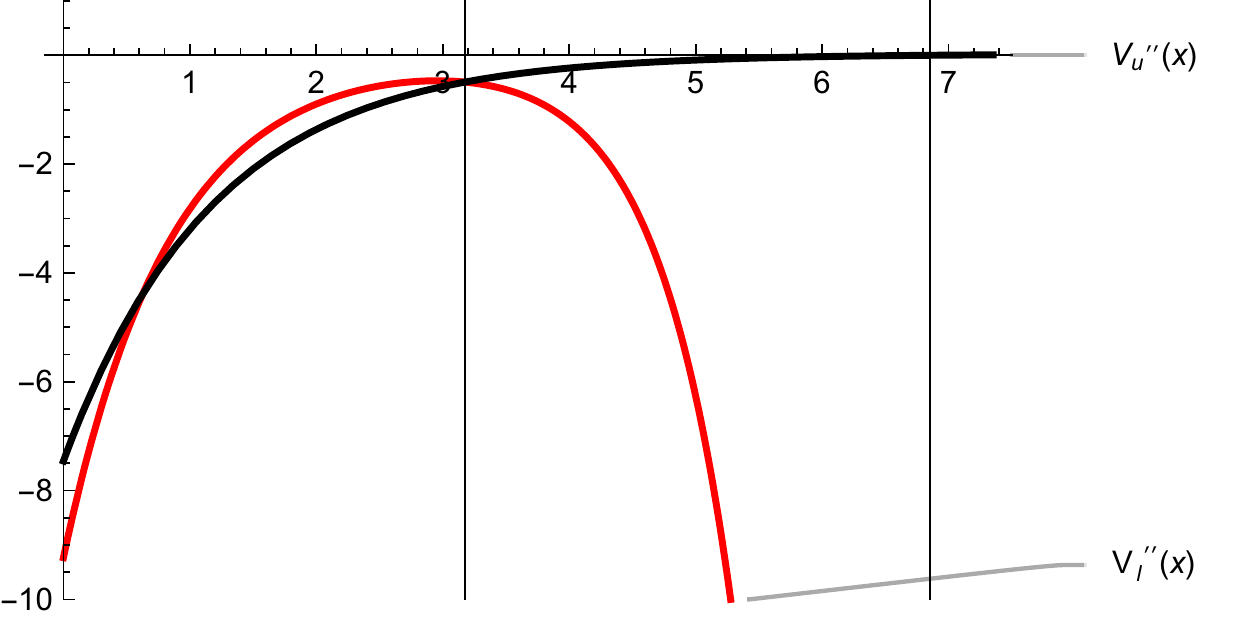}
\caption{2nd order smooth fit}\label{fig:2ndsmooth}
\end{minipage}
\end{figure}


\end{document}